\documentclass[11pt]{article}

\usepackage{amsmath}
\usepackage{amssymb}
\usepackage{graphicx}
\usepackage{float}
\usepackage{caption}
\usepackage{amsthm}
\usepackage{enumerate}
\usepackage{color}
\usepackage{authblk}
\definecolor{green}{rgb}{0,0.5977,0}
\usepackage[linesnumbered,algoruled,boxed,lined]{algorithm2e}
\usepackage{multirow}

\newcommand{\rr}{\mathbb R}

\newcommand{\ee}{\mathbb E}

\newcommand{\suchthat}{\ | \ }

\newcommand{\twocases}[4]{\begin{cases} #2 & #1 \\ #4 & #3 \end{cases}}
\newcommand{\threecases}[6]{\begin{cases} #2 & #1 \\ #4 & #3 \\ #6 & #5 \end{cases}}

\newcommand{\txt}[1]{\text{#1}}

\newcommand{\stext}[1]{\ \ \ \ \ \text{(#1)}}
\newcommand{\stextn}[1]{\\&\ \ \ \ \ \ \stext{#1}}

\newcommand{\lp}[3]{\begin{tabular}{l l}\textbf{#1} & \begin{tabular}{l}$#2$\end{tabular}\\\textbf{subject to} & \begin{tabular}{l l}#3\end{tabular}\end{tabular}}
\newcommand{\lpmax}[2]{\lp{maximize}{#1}{#2}}

\theoremstyle{plain}
\newtheorem{theorem}{Theorem}
 
\newtheorem{proposition}[theorem]{Proposition}

\theoremstyle{definition}

\numberwithin{theorem}{section}

\title{Cheap Talk in Bilateral Trade}
\author{Jamie Tucker-Foltz\thanks{Yale School of Management, \texttt{j.tuckerfoltz@yale.edu}} \qquad Richard Zeckhauser\thanks{Harvard Kennedy School, \texttt{richard\_zeckhauser@harvard.edu}}}

\begin{document}

\maketitle

\begin{abstract}
    A single seller offers one or more goods to a single buyer. The buyer's values and the seller's costs are private information. Each player has a commonly known prior over the other player's value or cost, supported on a finite set. What is the optimal selling mechanism?
    
    We argue that, despite this question's importance and apparent simplicity, prior work offers no satisfactory answer. If the seller simply chooses an optimal menu given her realized costs, she fails to exploit her informational advantage. At the other extreme, the optimal trade mechanism that satisfies IC/IR constraints for both parties fails in practice, as it conditions prices on the seller's unknown costs in an unenforceable way. The seller's realistic capabilities lie somewhere in between: she may leverage private information but lacks unlimited commitment power.
    
    To bridge this gap, we consider a solution concept built on the realistic assumption that the seller can commit to prices but nothing more. Similar---albeit technically distinct---solution concepts have been studied in the context of auctions with multiple buyers. Our concept proves surprisingly rich even with a single buyer. In our model, the buyer and seller engage in multiple rounds of cheap talk before the seller posts a menu of priced bundles. The buyer then purchases.
    
    We measure value as profit for the seller and consumer surplus for the buyer. We prove that, when there is only one good, such cheap talk cannot improve the welfare of either party. We then demonstrate that cheap talk \emph{can} be useful when there are (1) multiple goods with additive costs and values, (2) multiple units of a single good with constant marginal cost and diminishing marginal value, (3) interdependent values for a single good, or (4) repeated play of a one-good game. We also show that multiple rounds of communication can yield strictly higher expected profit than a single round.
    
    Conceptually, these results show that in any extension beyond the canonical setting of one seller, one buyer, and one good, cheap talk creates value in bilateral trade. We discuss how realistic factors beyond our stripped-down model combine with cheap talk to enhance this value even further.
\end{abstract}

\section{Introduction}\label{secIntro}

Some items are sold with prices posted by sellers. Others, particularly high-priced items such as cars, houses, and corporations, typically include bargaining. Bargaining involves a mix of statements and offers. A critical question explored here is: when can a player benefit by making a statement or statements? Statements by themselves are costless. When they are also non-binding and unverifiable, they are often referred to as cheap talk. That leads to the question: \emph{when} and \emph{how} can cheap talk be valuable?

We focus on the well-studied bilateral trade setting, with a buyer and seller potentially transacting to sell a good or goods. Each player has private information that is a realization from a common value prior. That information is value for the buyer and cost for the seller. It is well known that when private knowledge is involved, full efficiency is lost. The question becomes: what is the best the pair can realistically hope to achieve?

Our answer to the \emph{when} question is that cheap talk is useless in the most basic setting involving a single good, but yields value once even the smallest of steps are taken into complexity. We answer the \emph{how} question by illustrating cheap talk at work in a variety of contexts.

Cheap talk in negotiation suffers from Cassandra's curse. Apollo gave Cassandra the power of prophecy but, when she spurned him, cursed her so that no one would believe her. So too, a negotiator may possess genuine private information and communicate it truthfully, but there is an ever-present incentive to misrepresent private information. Listeners can never fully trust the message; thus, full efficiency remains out of reach. 

A prime goal of this paper, however, is to demonstrate that cheap talk can significantly raise efficiency by enabling the buyer and seller to reach better deals, that is to secure mutually beneficial transactions that would be lost if sellers simply posted prices---possibly a complex menu of prices---without any communication.

\subsection{Outline of Contributions}\label{subContributions}

Before formally defining our model, we begin in Section~\ref{secExamples} with a series of examples that illustrate the main theses of this paper:
\begin{itemize}
	\item Section~\ref{subExampleNoGains} informally demonstrates why cheap talk fails to create value in the simplest possible market interaction, one where there is a single good being sold with a single message in a single round. We later strengthen this intuition into a strong impossibility theorem for arbitrary communication protocols.
	\item Section~\ref{subExampleBuyerWelfareImprovement} provides a concrete, believable example where cheap talk yields positive surplus under the following minimally-modified condition: The seller holds a binary piece of information that affects the buyer's value which must be disclosed to the buyer prior to sale. The buyer can quadruple his ex ante consumer surplus by making a voluntary, cheap talk declaration prior to learning this piece of information.
	\item Section~\ref{subExampleMechDesign} considers an extended example with multiple goods and independent costs/values, viewed from the perspective of the seller. What is the seller's optimal mechanism? This question has been widely studied in microeconomic theory. We argue, however, that existing solution concepts fail to provide a satisfactory answer in the context when the seller holds private information. We explain why the revelation principle does not hold. We then describe a dynamic mechanism that leverages cheap talk to extract a greater expected profit at the expense of the buyer.
\end{itemize}

In Section~\ref{secLitReview}, we connect our work to the broader literature on cheap talk, which has largely focused on abstract sender-receiver games and ignored bilateral trade. The literature on informed principal problems and commitment issues in mechanism design is more relevant to our setting. We explain how our cheap talk solution concept differs from the recently proposed axiom of \emph{credibility} from Akbarpour and Li~\cite{CredibleTrilemma}.

Our basic model is formally defined in Section~\ref{secModel}. At a high level, there is one seller, one buyer, and a finite number of goods to be sold. The seller's costs and buyer's values are private information, drawn independently from finitely supported  and commonly known distributions. A cheap talk game involves players taking turns sending messages for a finite number of rounds, after which the seller posts a menu. Each item on the menu consists of a bundle of goods, each with an associated price. Buying zero goods is required to be an option. Then the buyer chooses one option from the menu.

Section~\ref{secImpossibility} presents Theorem~\ref{thmImpossibility}, our main impossibility theorem showing that cheap talk cannot yield value in the simplest context. As fundamental as this statement is, it surprisingly does not seem to exist in prior literature. The difficulty of the proof is also surprising. The theorem states that cheap talk cannot be beneficial in the basic model when there is only one good and one message. In more detail, any equilibrium of a cheap talk game can be replaced by a new equilibrium where signals are uninformative. In the new equilibrium, both players receive the same payoffs. The proof proceeds by a careful doubly-nested induction argument from the lowest-value type up to the highest, invoking the incentive constraints to identify a specific signal $X$ that is the best choice for all buyer types in each subgame corresponding to the final round of communication. Using the known result that the optimal price is a monotone function of the seller's cost, we show that, if sellers assume the buyer sent $X$ in each subgame, we get a payoff-equivalent equilibrium.

Section~\ref{secPositive} formally presents extensions of this basic setting on a variety of dimensions, including the ones informally described in Section~\ref{secExamples}, along with some new extensions:
\begin{itemize}
	\item Multiple goods, with additive values and no changes to the basic model.
	\item Multiple units of single good with constant marginal cost for the seller and diminishing marginal value for the buyer.
	\item Interdependent values for a single good, which we show reduces to the setting of multiple goods.
	\item Repeated play of a one-good game.
\end{itemize}
We describe our computational approach to verify the existence of equilibria where cheap talk yields value in these various settings. These results stand in stark contrast to Theorem~\ref{thmImpossibility}. We also construct an example where multiple rounds of cheap talk provably yield higher profits than in any pure strategy equilibrium of a 1-round cheap talk game.

Section~\ref{secDiscussion} concludes with a discussion of factors beyond the scope of our model, including human behavioral tendencies, that can further enhance the value of cheap talk in practice.

\section{Examples}\label{secExamples}

When does cheap talk yield better deals? We proceed from extremely simple to more complex situations.

\subsection{One Good, One Play, One Message: No Gains from Cheap Talk}\label{subExampleNoGains}

\subsubsection{Distribution on Buyer Value; Seller Cost Fixed}

The simplest sale situation arises when a seller is seeking to sell a good with a fixed cost $c$ to a single buyer whose value for the good, $v$, is unknown. However, the distribution of the buyer's value, $f(v)$, and its cumulative distribution, $F(v)$, are common knowledge. The buyer purchases if his value is greater than or equal to the price $p$. The seller sets a price $p$ to maximize her profit if the item sells times the likelihood of a sale, namely
\[
(p - c)[1 - F(p)].
\]
Here, $1 - F(p)$ is a typical demand curve, representing the probability that the buyer’s value exceeds $p$.

In this context, neither the buyer nor the seller could gain by making a statement. The seller cannot convey any new information. For the buyer, as each potential signal deterministically elicits an optimal best response for the seller,\footnote{This argument is not fully rigorous because there is a possibility that multiple best responses exist, in which case the seller could randomize over them. We make a much stronger version of this argument precise in Theorem~\ref{thmImpossibility}.} all buyer types should choose the signal yielding the lowest price. Thus, the signal is meaningless to the seller; she might as well ignore it.

\subsubsection{Distribution on Buyer Value; Distribution on Seller Cost}

The seller’s cost may also have a distribution, with that distribution being common knowledge. In the posted-price situation, the seller learns her cost before she posts, and she never posts a price below her cost. A numerical example reveals the potential for cheap talk to be beneficial.

Let the buyer’s value distribution be
\[
f(v) = \threecases{\txt{with probability } 0.05}{\$100}{\txt{with probability } 0.30}{\$50}{\txt{with probability } 0.65}{\$10}
\]
and the seller’s cost distribution be
\[
g(c) = \twocases{\txt{with probability } 0.40}{\$40}{\txt{with probability } 0.60}{\$0}.
\]

For each price and cost realization, the seller’s expected profit is
\[
\Pr[\text{sale}] \times (p - c).
\]
For example, if $c = \$40$ and $p = \$50$, sale occurs when $v \ge \$50$, so expected profit is
\[
\pi = 0.35 \times (\$50 - \$40) = \$3.50.
\]

\begin{table}[h!]
	\centering
	\begin{tabular}{|r|c c c|}
		\hline
		& $p=\$100$ & $p=\$50$ & $p=\$10$ \\
		\hline
		$\Pr(\text{sale})$ & 0.05 & 0.35 & 1.00 \\
		\hline
		$\pi$ if $c=\$40$ & \$3.00 & \textbf{\$3.50} & -- \\
		$\pi$ if $c=\$0$  & \$5.00 & \textbf{\$17.50} & \$10.00 \\
		\hline
	\end{tabular}
	\caption{Seller as a function of posted price and realized cost. Bolded entries indicate optimal prices.}
	\label{tabSellerProfits}
\end{table}

With either cost type, the seller posts $p = 50$. Expected seller profit is
\[
0.40(\$3.50) + 0.60(\$17.50) = \$11.90.
\]
Equivalent expected-value calculations show that buyer consumer surplus is $\$2.50$.

\subsubsection{Can Cheap Talk Help?}

The buyer proposes cheap talk that would provide benefit to both players. He suggests the following strategy:
\begin{quote}
	``If my value is $\$100$ or $\$10$, I will say $A$. You will price at $\$100$ when $c=\$40$ and at $\$10$ when $c=\$0$. If my value is $\$50$, I will say $B$, and you will price at $\$50$.'' 
\end{quote}
These prices are optimal for the seller if the buyer adheres to the proposal. The expected values are shown in Table~\ref{tabCheapTalkOutcomes}. When the buyer adheres, payoffs are higher for both players than under no communication.

\begin{table}[h!]
	\centering
	\begin{tabular}{|l|c c c|}
		\hline
		& No Communication & Buyer Adheres & Buyer Defects \\
		\hline
		Seller Profit & 11.90 & 15.60 & 7.20 \\
		Consumer Surplus & 2.50 & 2.70 & 9.90 \\
		Total Surplus & 14.40 & 18.30 & 17.10 \\
		Deadweight Loss & 3.90 & 0 & 1.20 \\
		\hline
	\end{tabular}
	\caption{Expected outcomes under no communication, proposed cheap talk with adherence, and deviation by the buyer with $v=50$.}
	\label{tabCheapTalkOutcomes}
\end{table}

However, cheap talk fails here, given the potential to misrepresent. The buyer with $v = \$50$ has a strong incentive to state $A$ rather than the proposed $B$. Given $B$, his consumer surplus is zero. Under the proposed pricing, if he states $A$, his expected surplus is
\[
0.60 \times (\$50 - \$10) = \$24.
\]
Were the seller to be fooled, the payoffs in the third column of Table~\ref{tabCheapTalkOutcomes} would apply.

We posit seller sophistication. The seller could easily anticipate that all buyer types are sending the same message and thus learns nothing from cheap talk. She therefore ignores the buyer’s statement.

Cheap talk would bring value if the buyer could somehow guarantee adherence to his proposal,\footnote{If the buyer could commit adherence, his optimal strategy would state A for $v = \$100$ and $v = \$10$, and would state A with probability $0.375$ and B with probability $0.625$ for $v = \$50$.} for example via a third-party enforcer or a repeated-play environment. Either arrangement violates our current assumptions.

\subsubsection{Why Cheap Talk Fails}\label{sssCheapTalkFails}

The failure of cheap talk here reflects a general principle: cheap talk cannot simultaneously reduce deadweight loss and leave both parties at least as well off as in the status quo. Thus, any informative cheap talk that reduces deadweight loss must harm at least one party. The harmed party will simply refuse to listen.

The heuristic logic is as follows. Deadweight loss arises from excluded low-value buyers. In the absence of communication, types with $v < p^*$ (the optimal posted price) do not trade, even though trade would be efficient when $c$ is low. Reducing deadweight loss requires lower prices for low-value types. For $v=\$10$ buyers to trade, the seller must sometimes charge $p=\$10$. However, lower prices attract mimicry from all buyer types. In particular, the $v=\$50$ buyer earns zero consumer surplus in the status quo, paying exactly his value. He has nothing to lose by sending any message that might induce a lower price, and has everything to gain. This inevitable misrepresentation ultimately causes the equilibrium to unravel.

While this reasoning is convincing, at least in this particular example, it is not obvious a priori that there is  no alternative pricing system that is simultaneously incentive compatible for both buyer and seller. One of our main contributions, Theorem~\ref{thmImpossibility}, formally demonstrates this impossibility.

\subsection{Buyer Gains from Cheap Talk}\label{subExampleBuyerWelfareImprovement}

Consider now a situation where the seller holds a piece of private information affecting the buyer's value. The motivating example is a new car showroom, a setting in which verbal statements---namely cheap talk---are often made to influence pricing.

There are two types of buyers: \emph{enthusiasts} and \emph{sticklers}, in equal proportions. They differ along two dimensions. Enthusiasts have high willingness to pay and care only modestly about color. Sticklers are thrifty and will only buy the color they prefer. Members of each type are equally likely to prefer a red or a blue car. The dealer’s cost is always \$30.

The buyer population is described in Table~\ref{tabTwoGoodsValues}.

\begin{table}[h!]
	\centering
	\begin{tabular}{|l|c c c c|}
		\hline
		& Stickler (Red) & Stickler (Blue) & Enthusiast (Red) & Enthusiast (Blue) \\
		\hline
		Value if Red & \$40 & \$0 & \$48 & \$46 \\
		Value if Blue & \$0 & \$40 & \$46 & \$48 \\
		\hline
	\end{tabular}
	\caption{Buyer types, probabilities, and valuations for red and blue cars.}
	\label{tabTwoGoodsValues}
\end{table}

The dealer is known to have exactly one car. Its color is not known to the buyer. In the no-communication case, the seller announces the car’s color and posts a price. With cheap talk, the buyer may first make a statement.

\textbf{No communication (seller has red car).}
The seller considers prices \$48, \$46, and \$40. At \$48, she sells only to a red enthusiast, yielding
\(
\ee[\pi] = \tfrac14 \times (\$48 - \$30) = \$4.5.
\)
At \$46, she sells to either enthusiast, yielding
\(
\ee[\pi] = \tfrac12 \times (\$46 - \$30) = \$8.
\)
At \$40, she sells to either enthusiast and the red stickler, yielding
\(
\ee[\pi] = \tfrac34 \times (\$40 - \$30) = \$7.5.
\)
Hence, she prices at \$46. Only the red enthusiast earns consumer surplus:
\(
\ee[CS] = \tfrac14 \times (\$48 - \$46) = \$0.5.
\)

\textbf{Cheap talk (seller has red car).}
If the buyer announces ``I prefer red,'' the seller faces a pool consisting of half sticklers with value \$40 and half enthusiasts with value \$48. At price \$48, she sells only to the enthusiast, yielding
\(
\ee[\pi] = \tfrac12 \times (\$48 - \$30) = \$9.
\)
At price \$40, she sells to both types with certainty, yielding
\(
\ee[\pi] = 1 \times (\$40 - \$30) = \$10.
\)
Thus, she prices at \$40.

If the buyer announces ``I prefer blue,'' the seller faces a pool consisting of half sticklers with value \$0 and half enthusiasts with value \$46. At price \$46, she sells only to the enthusiast, yielding
\(
\ee[\pi] = \tfrac12 \times (\$46 - \$30) = \$8.
\)
Lower prices to accommodate the stickler are unprofitable, so the seller prices at \$46.

Overall, the seller’s expected profit with cheap talk is
\(
\ee[\pi] = \tfrac12 \times \$10 + \tfrac12 \times \$8 = \$9.
\)

Sticklers receive zero consumer surplus. A red enthusiast whose color matches earns \$48 - \$40 = \$8, which occurs with probability $\tfrac14$. A blue enthusiast pays \$46 for a car worth \$46, earning zero surplus. Hence,
\(
\ee[CS] = \tfrac14 \times \$8 = \$2.
\)

The analysis is symmetric when the seller has a blue car.

\begin{table}[h!]
	\centering
	\begin{tabular}{|r|c c|} 
		\hline
		& Seller profit (\$) & Buyer surplus (\$) \\
		\hline
		No-communication & 8 & 0.5 \\
		Cheap talk & 9 & 2 \\
		\hline
	\end{tabular}
	\caption{Expected payoffs with and without cheap talk.}
	\label{tabTwoGoodsComparison}
\end{table}

Both parties benefit from cheap talk. The seller’s expected profit rises from \$8 to \$9, while the buyer’s expected surplus quadruples from \$0.5 to \$2.

\subsection{Optimal Mechanism Design}\label{subExampleMechDesign}

The potential for utility-improving cheap talk raises intriguing questions for mechanism design. An example where the seller holds private information illustrates the subtleties of her pricing problem.

There are two goods (numbered 1 and 2), two possible seller types (each occurring with equal probability $\frac12$), and three possible buyer types (each occurring with equal probability $\frac13$). The possible seller costs and buyer values are shown in Table~\ref{tabHierarchyInstance}.

\begin{table}[h!]
	\centering
	\begin{tabular}{|r|c c|c c c|} 
		\hline
		& \multicolumn{2}{c|}{Seller cost (\$)} & \multicolumn{3}{c|}{Buyer value (\$)}  \\		
		& Type $A$ & Type $B$ & Type 1 & Type 2 & Type 3 \\
		\hline
		Good 1 & 168 & 270 & 492 & 546 & 84\\
		
		Good 2 & 246 & 30 & 330 & 210 & 306\\
		\hline
	\end{tabular}
	\caption{A simple market with multiple goods and two-sided private information. Each seller type is equally likely and each buyer type is equally likely.}
	\label{tabHierarchyInstance}
\end{table}

What is the seller's optimal mechanism to sell the goods? Without the uncertainty over seller types, there is a well-known solution to this problem. The revelation principle holds that any extensive-form game the seller could possibly design is equivalently implemented by first having the buyer directly reveal his type, then having the seller simulate what would have happened and announcing the outcome. Implicit is the assumption that the seller can commit to faithfully simulating the original game. However, in the context where a single buyer holds the only private information, only a very weak---and therefore quite plausibly realistic---form of commitment is actually required from the seller. This is because a direct-revelation mechanism is equivalent to the seller posting a menu of prices for various bundles of goods, and the buyer choosing one option. Thus, the seller needs only to commit to not haggling over prices once they have been posted on the menu, the widely accepted social norm in many markets.

The seller can surely post an optimal menu given her realized costs. We call this the \emph{no-communication menu}. Finding the optimal menu for each realization separately can be accomplished by linear programming, as elaborated in Section~\ref{subLP}. Combining the menus, the solution has the structure shown in Table~\ref{tabHierarchyMechanism}, specializing $X = 798$ and $Y = 756$. In other words, when the seller has Type $A$, she offers the menu in the first row: Good 1 for \$546, Good 2 for \$306, or both goods for \$798. When the seller has Type $B$, she offers Good 2 alone for \$306 or both goods for \$756. This strategy yields her an ex ante expected profit of
\$335. We refer to this figure as the \emph{no-communication profit}.

\begin{table}[h!]
	\centering
	\begin{tabular}{|r r|c c c|}
		\hline
		&& \multicolumn{3}{c|}{Buyer}\\
		&& Type 1 & Type 2 & Type 3 \\
		\hline
		\multirow{2}{*}{Seller}
		& Type $A$ & Both goods for $\$X$ & Good 1 for \$546 & Good 2 for \$306 \\
		& Type $B$ & Both goods for $\$Y$ & Both goods for \$756 & Good 2 for \$306 \\
		\hline
	\end{tabular}
	\caption{The optimal selling mechanism for the instance from Table~\ref{tabHierarchyInstance}, where $X$ and $Y$ depend on the solution concept.}
	\label{tabHierarchyMechanism}
\end{table}

Can the seller do any better? This solution has a severe shortcoming: it fails to leverage the seller's informational advantage: The buyer does not know the seller's cost. By simply posting her optimal menu, the seller might as well announce her cost to the buyer at the start of the game, effectively surrendering her advantage while gaining nothing in return. The mechanism design literature suggests a more powerful framework for handling settings where multiple players have private information. Deploying the revelation principle, it entreats all players to reveal their private information ``to the mechanism''. It then determines the allocation of the goods and the monetary transfers between players and the mechanism. Any mechanism satisfying the constraints of \emph{individual rationality (IR)} and \emph{incentive compatibility (IC)} is deemed to be valid. Thus, it optimizes over this class to pick the best IR/IC mechanism for the seller. We call this the \emph{commitment game}. An optimal mechanism for the instance at hand happens to still take the form of Table~\ref{tabHierarchyMechanism}, but it charges higher prices to the Type 1 buyer: $X = 822$ and $Y = 774$.

Note that this mechanism is \emph{not} implementable by having the seller simply post prices as in the no-communication game. Problematic deviations arise for either menu:
\begin{enumerate}
	\item\label{itmDeviationA} If the seller posts the Type $A$ menu, then a Type 1 buyer would prefer the Type 3 option to get Good 2 for \$306, taking a surplus of \$24 rather than \$0.
	\item\label{itmDeviationB} If the seller posts the Type $B$ menu, then a Type 1 buyer would prefer the Type 2 option to get both goods for \$756, taking a surplus of \$66 rather than \$48.
\end{enumerate}
In the commitment game, however, the buyer does not know which of these menus the seller will be using when he reveals his type. Therefore, truthful revelation turns out to be incentive compatible, as each menu requires a different deviation. Instead of the no-communication profit of \$335, the seller now obtains an ex ante profit of \$342, which we call the \emph{commitment profit}.

As the name suggests, we view this prediction skeptically because it requires an unrealistic level of commitment power from the seller. Suppose that the seller simulates the game according to the revelation principle. If the buyer reveals he is Type 1, the seller is supposed to charge \$822 if she has Type $A$ and \$774 if she has Type $B$. Clearly, a rational seller will always claim to have Type $A$ regardless, charging \$822, which is the maximum amount a Type 1 buyer is willing to pay. The game unravels.

We thus return to our fundamental starting question: What is the seller's optimal profit? Is it merely the no-communication profit of \$335, or the commitment profit of \$342? We argue that the correct answer in this case turns out to lie strictly in between; \$339 to be exact. This is the result of a dynamic mechanism whereby the seller asks the buyer which good he values more and then posts a menu of prices contingent on the buyer's response and her own private cost. The buyer's response is cheap talk; thus we refer to this game as a \emph{cheap talk game} and the optimal profit of \$339 as the \emph{cheap talk profit}. The ultimate allocation and transfers work out to again take the same form as in the no-communication and commitment games, with $X = 822$ and $Y = 756$. The comparison of these solution concepts is summarized in Table~\ref{tabHierarchyComparison}.

\begin{table}[h!]
	\centering
	\begin{tabular}{|r|c c c c|} 
		\hline
		& Value of $X$ & Value of $Y$ & Seller profit (\$) & Buyer surplus (\$) \\
		\hline
		No-communication & 798 & 756 & 335 & 15 \\
		Cheap talk & 822 & 756 & 339 & 11 \\
		Commitment & 822 & 774 & 342 & 8 \\
		\hline
	\end{tabular}
	\caption{Comparison of prices and welfare in each solution concept for the instance from Table~\ref{tabHierarchyInstance}.}
	\label{tabHierarchyComparison}
\end{table}

There are several potential deviations that must be checked in this game, but none of them turns out to be profitable. The buyer does not wish to lie, and the seller's menu is optimal given that the buyer is truthful. The most crucial observation is that Deviation (\ref{itmDeviationA}) above is ruled out because the Type 3 option will not be on the menu, as a Type 1 buyer has already told the seller he does not have Type 3. But if $Y > 756$, we cannot also rule out Deviation (\ref{itmDeviationB}), which is why although the commitment game raises the value of $Y$, the cheap talk game does not.

In contrast to the example from Section~\ref{subExampleBuyerWelfareImprovement}, we note that, in this case, cheap talk does not result in a Pareto improvement for the two players. The seller's expected utility increases by \$4, the buyer's utility decreases by \$4, and the deadweight loss stays constant. Since this mechanism relies on the buyer's participation, a natural objection is that the buyer should simply stand mute. However, implicit in the framing of our mechanism design question is an assumption that the buyer cannot commit to silence. The seller can ask the buyer, ``Which good you value more? If you do not answer, I will assume you prefer Good 1.'' The seller's belief is credible at equilibrium, and the buyer will have to answer. In this sense, ``not communicating'' is not an option after all. This is a quirk not only of our model but more generally of dynamic mechanisms: players other than the designer can effectively be coerced into playing games they do not wish to play.

\section{Relation to Previous Literature}\label{secLitReview}

The notion that cheap talk can improve equilibrium outcomes is well-known and widely studied. The seminal work of Crawford and Sobel~\cite{CrawfordSobel} considers the canonical setting where a sender must communicate information to a Receiver, who then takes an action. This is similar to the more widely-studied Bayesian persuasion framework~\cite{NormalForm}, but the Sender lacks commitment power. Recent work has studied computational aspects of this one-way information design problem~\cite{AlgorithmicCheapTalk}. There has also been much prior work in economics on cheap talk in normal-form games where both parties take actions~\cite{NormalForm}, as well as models of cheap talk in specific economic settings~\cite{CheapTalkEcon1, CheapTalkEcon2, EfficientCompetition}.

However, none of these results applies to one of the most basic problems in economic theory: an interaction between a single seller and a single buyer. The conceptual finding of Riley and Zeckhauser~\cite{RileyZeckhauser} (also established by Myerson~\cite{MyersonAuction}) holds that, when there is one buyer, one seller, and one good, the seller should not engage in ``haggling.'' Rather, they should commit to a single price for the good. Notably, this is only shown for the setting where the only private information is on the side of the buyer. Our Theorem~\ref{thmImpossibility} generalizes this result to the case where the seller has private costs.

With two-sided private information, an approach taken in many prior works would view our problem as that of designing a \emph{bilateral trade mechanism}. In this model, a third party solicits the private information of each individual. The classic Myerson-Satterthwaite Theorem~\cite{MyersonSatterthwaite} shows that not all beneficial trades can be facilitated. (However, if the individual rational requirement can be relaxed to apply only \textit{ex ante}\textit{}\ the expected externality mechanism achieves full efficiency. In a negotiation, if the seller is paid the ``expected externality'' of their offer---the difference, if positive, between the buyer’s valuation and the seller’s cost---it induces both parties to report their true valuation. A sale results, e.g., at a price half way between the offers, whenever a mutually beneficial transaction exists~\cite{ExpectedExternalities}.)

Taking the seller's perspective, the question of optimal mechanism design in our setting sits within the \emph{informed principal} literature, where the designer possesses private information. We follow the framework set out in the seminal paper by Myerson~\cite{MyersonInscrutability}, where we assume without loss of generality that the mechanism does not depend on the seller's type. A large body of work has studied such mechanism design problems, most notably the papers by Maskin and Tirole~\cite{MaskinTiroleInformedPrincipal1, MaskinTiroleInformedPrincipal2}. More relevant to our bilateral trade context is the work of Yilankaya~\cite{Yilankaya}, which builds on a result of Williams~\cite{Williams}, showing a similar impossibility result as our Theorem~\ref{thmImpossibility}. However, they operate utilizing a stylized, continuous model and the theorem requires the regularity condition that both players' virtual values are increasing. We require no such assumptions. To the best of our knowledge, no prior work has shown the general impossibility of cheap talk in the classic Myerson-Satterthwaite setting with one buyer, one seller, and one good. Further afield, several works have looked at more complex information design problems with an informed seller and multiple buyers~\cite{EInformedPrincipal3, EInformedPrincipal4, EInformedPrincipal1}.

An alternative perspective, as we took in Section~\ref{subExampleMechDesign}, starts from the classic mechanism design setting and imposes additional constraints limiting the ability of parties to commit. Most relevant is the notion of \emph{credibility} from Akbarpour and Li~\cite{CredibleTrilemma}. Informally, a mechanism is credible if some buyer could detect any deviation by the seller. The canonical example of a non-credible mechanism is a sealed-bid second-price auction among multiple buyers. The seller can pretend to the winning buyer that another buyer bid just slightly lower, thereby extracting greater revenue.

In the context of a single buyer, credibility becomes a meaningful constraint if (and only if) the seller has private information. The seller's costs are unverifiable, so the seller cannot credibly condition any outcome on them. Any cheap talk game (formally defined in Section~\ref{secModel}) is a credible mechanism, but \emph{not} necessarily vice-versa. For instance, in the example from Section~\ref{subExampleMechDesign}, a variant of the optimal commitment mechanism can be implemented credibly. The only twist is that the Type 1 buyer is always charged \$798 for both goods regardless of the seller's cost (rather than being a coin flip between \$822 or \$774).\footnote{If we drop the requirement that seller's allocation satisfies ex post IR, then this transformation shows that \emph{any} commitment mechanism can be made credible. However, we do not know if the combination of credibility and ex post IR constraints can reduce the seller's optimal utility relative to their commitment profit. Computational experiments failed to find any such examples.} The Type 3 buyer is still always charged \$306. The Type 2 buyer, however, does not know what they are going to be charged, nor what goods they will receive. That depends on the seller's costs. While this mechanism does not leave room for the seller to profitably lie, it nevertheless requires a level of commitment beyond simply posting a menu. The two players must contractually agree that one of two types of transactions will take place (Good 1 for \$546 or both goods for \$756) at the choice of the seller. If the buyer were allowed to back out of this obligation and walk away with no goods and no payment, the overall mechanism would not be incentive compatible. The Type 1 buyer could pretend to be Type 2 and walk away if the seller is Type A. Thus, for contexts where both players lack commitment power, we believe that our cheap talk model better grasps reality than the credibility formulation of Akbarpour and Li~\cite{CredibleTrilemma}.

To the best of the authors' knowledge, there is only one prior work, by Banchio, Skrzypacz and Yang \cite{CredibleHiddenCost}, that has studied credibility in the context where the seller has private information. As with the rest of the credibility literature, they focus on settings involving multiple buyers. Whereas Akbarpour and Li identify the first-price auction as the unique optimal auction format that is static and credible, Banchio et al.~\cite{CredibleHiddenCost} show that, when the seller has hidden costs, first price auctions lack credibility. In fact, no optimal auction format can be simultaneously static and credible. Our work differs substantially, given our focus on the case of a single buyer (albeit with extensions in this context to previously unstudied multiple features of the market, such as multiple goods and interdependent values).

Alternative formulations of limited commitment have been proposed. Among the most relevant, McAdams and Schwarz~\cite{CredibleSales} studies an arguably weaker kind of seller who cannot even commit to closing a sale when desirable. Bester and Strausz~\cite{CommitmentAndRevelationPrinciple} consider a very general model where there is an arbitrary space of outcomes with a component to which the mechanism designer cannot commit. The commitment issues our work considers stem from a different source: the inability to contract outcomes based on private information held by the designer.

Our work connects as well to a large literature on market segmentation: A seller wishes to engage in third-degree price discrimination by offering different prices to different segments of the market. A highly influential paper by Benjamin, Brooks, and Morris~\cite{BBM} characterizes the buyer/seller welfare combinations attainable by an intermediary who knows the buyer's value selectively revealing information to the seller, who has a known cost. Hidir and Vellodi~\cite{EBuyerCommunicates} examine what happens when the buyer's value is private. In place of an omniscient intermediary, the buyer communicates via cheap talk. However, unlike in our paper, their setting constrains the seller to offer a single price for a single good, chosen out of a set of possible goods. Our setting allows for arbitrary selling mechanisms. The usefulness of cheap talk in our setting is driven by the private seller costs, rather than a constraint on the kinds of selling mechanisms allowed. 

There has also been work on improving market outcomes via \emph{verifiable} information disclosure, where the buyer can provably reveal information about their values. This is similar to our approach in that neither player has commitment power, but the potential for verification leads to significantly different results. Halpern, Kehne, and Tucker-Foltz~\cite{BuyersReveal} study the extent to which voluntary disclosure can improve buyer welfare and lead to Pareto efficient outcomes in settings with multiple goods. While voluntary disclosure is always possible in the context of a single good~\cite{AliLewisVasserman}, that property is lost when there are multiple goods. Curiously, our main findings for cheap talk are precisely the reverse. Mao, Paes Leme, and Wang~\cite{ItcsBilateralTrade} study a bilateral trade model where both sides may disclose verifiable information, finding that the optimal communication protocol achieves fully-efficient trade.

Finally, there is also a literature falling between sender-receiver games and bilateral trade, in which the seller's private information affects the buyer's value, rather than the seller's cost. Koessler and Skreta~\cite{EInformedPrincipal2} consider a similar setting to our ``interdependent values'' model in Section~\ref{subInterdependentValues}. A key difference, however, is that we assume the seller must reveal the relevant information to the buyer's value before the good is sold. Both of these models can be cast as settings with independent values and multiple goods (see Proposition~\ref{proReduction}), but in Koessler and Skreta's model, the seller would essentially have full commitment power in the multi-goods game. Several other works examine various models of a seller strategically revealing information to a buyer, with a conceptual takeaway that having multiple dimensions (e.g., product attributes) makes cheap talk useful~\cite{EMultidimensional2, EMultidimensional3, EMultidimensional4, EMultidimensional5}. This echoes our result that cheap talk is useful when their are multiple goods, or multiple variants of the same good.

We note that the mechanics of our cheap talk model fall under the label of \emph{long} cheap talk~\cite{LongCheapTalk}, where multiple rounds of communication are allowed, and can be demonstrably useful~\cite{CheapTalkEcon2, EMultiRounds}. Our solution concept, described in the next section, presumes a fixed but unlimited number of rounds of messages, alternating between the two players, followed by a selling mechanism chosen by the seller.  This is the simplest and most realistic model of cheap talk as it exists in real markets---haggling. But we bar public randomness and simultaneous messages. That restriction weakens our impossibility result and strengthens our possibility results.

\section{Model and Solution Concept}\label{secModel}

We begin by describing the basic setting, which we will modify as needed when we consider extensions in Section~\ref{secPositive}. For any positive integer $n$, we denote $[n] := \{1, 2, \dots, n\}$.

We consider a class of games involving two players, a seller and a buyer. There are $m$ goods, indexed by the letter $j$ and numbered 1 through $m$. There are $n$ possible buyer types, indexed by the letter $i$ and numbered 1 through $n$, and $\ell$ possible seller types, indexed by the letter $k$ and numbered 1 through $\ell$. The distributions of these types are independent and common knowledge. For each buyer type $i \in [n]$, we denote their value as a vector $v^i \in \rr_{\geq 0}^m$, and for each seller type $k \in [\ell]$, we denote their cost as $c^k \in \rr_{\geq 0}^m$. Thus, the buyer value and seller cost for good $j$ are respectively denoted $v^i_j$ and $c^k_j$.

We only consider games with a finite number of rounds $R$ and a finite number of messages $M$. For convenience, we index rounds backward, so that round $r = R$ is the first communication, and $r = 1$ is the final communication before the seller posts a menu. Formally, the \emph{$(M, R)$-cheap talk game} proceeds as follows.

\begin{enumerate}
	\item\label{itmGameBuyerObserves} The buyer observes his type $i$. Equivalently, he observes his value vector $v^i \in \rr_{\geq 0}^m$.
	\item\label{itmGameSellerObserves} The seller observes her type $k$ and corresponding cost vector $c^k \in \rr_{\geq 0}^m$.
	\item\label{itmGameCheapTalk} A series of $R$ cheap talk rounds occurs, from round $r = R$ down to round $r = 1$. On each such round $r$:
	\begin{itemize}
		\item If $r$ is even, the seller sends a message to the buyer by publicly choosing a message from the set $[M]$.
		\item If $r$ is odd, the buyer sends a message to the seller by publicly choosing a message from the set $[M]$.
	\end{itemize}
	Note that, since $r = 1$ is odd, this means the buyer will always be the last to communicate. We adopt this convention because a communication from the seller to the buyer is useless on the final round, as the buyer's final action will just be picking the bundle that maximizes his utility.
	\item\label{itmGameSellerPostsMenu} The seller posts a menu of bundles of goods and associated prices
	$$((x^1, p^1), (x^2, p^2), \dots, (x^n, p^n)),$$
	where each $x^i \in [0, 1]^m$ and $p^i \in \rr_{\geq 0} \cup \{\infty\}$. It is without loss of generality to assume that the seller's optimal menu consists of at most $n$ options, one for each buyer type. If there were more possible options, at least one would be suboptimal for all buyer types, so could be eliminated. In the special case of $m = 1$ good, it is without loss of generality to assume the seller posts a single price, $p$, in which case we suppress the $x$ notation and superscripts. Allowing for a price of $\infty$ is for notational convenience, giving a convenient way for the seller to refuse sale. Any sufficiently large finite price would have the same effect.
	\item\label{itmGameBuyerAccepts} The buyer chooses an option $i' \in \{0, 1, 2, \dots, n\}$. If he chooses $i' = 0$, then payoffs are $(0, 0)$. Otherwise, the buyer and seller payoffs are
	$$u^B := -p^{i'} + \sum_{j = 1}^m x^{i'}_j \cdot v^i_j \hspace{1cm} u^S := p^{i'} - \sum_{j = 1}^m x^{i'}_j \cdot c^k_j.$$
\end{enumerate}
When $R = 0$, we refer to this game as the \emph{no-communication game}.

Our equilibrium concept is that of subgame-perfect Bayes-Nash equilibrium. In all of our positive results about the existence of utility-improving cheap talk games, we construct pure strategy equilibria. Of our two negative results, one of them (Theorem~\ref{thmImpossibility}) holds even for mixed-strategy equilibria, and the other (Theorem~\ref{thmMultiRounds}) considers only pure-strategy equilibria.

\section{An Impossibility Theorem For the Case of One Good}\label{secImpossibility}

Our main impossibility theorem essentially states that, when there is only one good, cheap talk is not helpful to either the buyer or seller. However, it does \emph{not} say that signals are ignored. It turns out that there are equilibria where the seller sets different prices depending on the signal, just not in a manner that improves expected payoffs. This feature complicates the statement and proof of our theorem:

\begin{theorem}\label{thmImpossibility}
	For any $M$ and $R$, in any mixed-strategy, subgame-perfect, Bayes-Nash equilibrium of an $(M, R)$-cheap talk game with $m = 1$ good, there exists an interim payoff-equivalent equilibrium where every message is uninformative (i.e., does not change the other player's beliefs).
\end{theorem}

For example, consider a simple (2, 1)-cheap talk game where the seller deterministically has cost 1, and the buyer has value 0 or 2, each with equal probability $\frac12$. The buyer could send signal $X$ when he has value 2 and $Y$ when he has value 0. The seller could optimally respond by setting $p = 2$ upon seeing signal $X$ and $p = \infty$ upon seeing signal $Y$. Thus the signals are not ignored. However, for the buyer, $X$ weakly dominates $Y$ across buyer types, and there is a payoff-equivalent equilibrium where the seller ignores the signal and just posts a price of 2.

The key claim in our proof is that, in any state of the game at round $r = 1$, there is always a signal like $X$ that behaves precisely as in this example with respect to any other signal $Y$. Deviating from $X$ to $Y$ is essentially a declaration from the buyer to the seller of the form, ``My type is so low that you should not even bother trying to sell to me if your cost is at least $c$.'' Identifying $X$ and establishing that it has this structure is the most technically challenging part of the proof, which requires a doubly-nested induction on the possible buyer types. Once established, a careful (though straightforward) analysis shows that we can modify the equilibrium so that the seller always pretends the buyer sent signal $X$, without breaking the sequential rationality or Bayesian consistency properties.

\begin{proof}
	Fix an equilibrium of an arbitrary $(M, R)$-cheap talk game. We will show that there is an interim payoff-equivalent equilibrium where the buyer's signal in the final round $r = 1$ is ignored. If $R \geq 2$, observe that the seller's final message in round $r = 2$ becomes irrelevant as well, given that the buyer's final communication is ignored. It thus follows inductively that all communication can be ignored.
	
	Assume the possible values for the $n$ types are numbered in order $0 \leq v^1 < v^2 < \dots v^n$. First we observe that, given any signal, it is without loss of generality to assume the seller always sets a price in $\{v^1, v^2, \dots, v^n, \infty\}$ (where $\infty$ is a stand-in for any value larger than the maximum cost or value). For given any price not in this set, if there is no larger $v^i$ that the buyer could have, then the seller could equivalently just set $p = \infty$ instead. Otherwise, the seller could increase the price to the next largest plausible $v^i$, which will yield greater revenue with the same probability.
	
	Fix an equilibrium $E_1$. Within this equilibrium, consider an arbitrary node in the extensive-form game tree on round $r = 1$, where the buyer holds a posterior probability distribution over the seller's possible types according to $E_1$. The buyer must choose which signal to send on this final round. For any signal $Z$ and value $v$, let $\Pr[p = v \suchthat Z]$ denote the probability that, when the buyer sends $Z$ at the given node in $E_1$, the seller sets price $p = v$. We claim that, for any signal $Y$, there is an index $i \in [n]$ such that:
	\begin{enumerate}[(i)]
		\item\label{itmSigYArgMax} For any index $1 \leq k < i$, signal $Y$ maximizes $\Pr[p = v^k \suchthat Z]$ over all possible signals $Z$, and this probability is nonzero.
		\item\label{itmSigYNothingAbove} For all indices $i < k \leq n$, $\Pr[p = v^k \suchthat Y] = 0$.
	\end{enumerate}
	
	It suffices to show that, given any signal $Y$ and index $i$, if $Y$ is not the signal that maximizes $\Pr[p = v^i \suchthat Z]$ over signals $Z$, then property (\ref{itmSigYNothingAbove}) holds for that index $i$. (Then we can just take $i$ to be the first index where this happens, and property (\ref{itmSigYArgMax}) will hold as well.) We prove this statement holds by induction on $i$. So let $i \in [n]$ be given, and assume it is true for all strictly smaller values (if there are any). To establish that property (\ref{itmSigYNothingAbove}) holds, we proceed by induction on $k$; i.e., fix any $k$ where $i < k \leq n$, and suppose property (\ref{itmSigYNothingAbove}) holds for all indices strictly between $i$ and $k$ (if there are any). Note that we have two inductive hypotheses, one for the index $i$ and one for the index $k$.
	
	Consider what happens if a buyer with value $v^k$ sends signal $Y$ versus some other signal $Z$ maximizing $\Pr[p = v^i \suchthat Z] > 0$. By the inductive hypothesis for index $i$, we know that $Z$ must maximize $\Pr[p = v^j \suchthat Z]$ for each index $1 \leq j < i$, for otherwise we would have $\Pr[p = v^k \suchthat Z] = 0$. In particular, this means that for each index $1 \leq j < i$, the probability $p = v^j$ is at least as high if the buyer sends signal $Z$ as if they send signal $Y$. For $j = i$, we have a strict inequality. Thus, we may bound the expected utility of the buyer towards sending signal $Z$ as
	\begin{align*}
		\ee[u^B \suchthat Z] &= \sum_{j = 1}^{i} \Pr[p = v^j \suchthat Z] (v^k - v^j) + \sum_{j = i + 1}^{k - 1} \Pr[p = v^j \suchthat Z] (v^k - v^j)\\
		&> \sum_{j = 1}^{i} \Pr[p = v^j \suchthat Y] (v^k - v^j) + \sum_{j = i + 1}^{k - 1} \Pr[p = v^j \suchthat Z] (v^k - v^j)\\
		&\geq \sum_{j = 1}^{i} \Pr[p = v^j \suchthat Y] (v^k - v^j) + \sum_{j = i + 1}^{k - 1} (0)(v^k - v^j)\\
		&= \sum_{j = 1}^{i} \Pr[p = v^j \suchthat Y] (v^k - v^j) + \sum_{j = i + 1}^{k - 1} \Pr[p = v^j \suchthat Y](v^k - v^j) \stextn{by the inductive hypothesis for index $k$}\\
		&= \ee[u^B \suchthat Y].
	\end{align*}
	Hence, the buyer with value $v^k$ strictly prefers to send signal $Z$, so he will never send signal $Y$. At equilibrium, the seller will therefore infer from a $Y$ signal that the buyer does not have value $v^k$. It makes no sense for the seller to set a price of $p = v^k$ by the same logic as before. She should either raise the price to $\infty$ or the next plausible buyer value. Thus, $\Pr[p = v^k \suchthat Y] = 0$, so by induction on $k$, property (\ref{itmSigYNothingAbove}) holds for all $i < k \leq n$.
	
	Thus, for any signal $Y$, we are always able to find some index $i$ satisfying properties (\ref{itmSigYArgMax}) and (\ref{itmSigYNothingAbove}). We denote the largest such index $i$ by $f(Y)$. We define $X$ to be any signal maximizing $f(X)$, breaking ties in favor maximizing the probability of $v^{f(X)}$. (Note that $X$ is defined with respect to a specific state within the game under equilibrium $E_1$. In what follows, we will abuse notation and refer to ``signal $X$'' as whatever message that corresponds to for the given state, which will be clear from context.) We know that, for any other signal $Y$, the distributions over prices given signals $X$ and $Y$ (where the randomness is over the seller's cost) have the following structure. For $1 \leq i < f(Y)$, we have $\Pr[p = v^i \suchthat X] = \Pr[p = v^i \suchthat Y]$, then for $i = f(Y)$ we have an inequality $\Pr[p = v^i \suchthat X] \geq \Pr[p = v^i \suchthat Y]$, and finally for all $f(Y) < i \leq n$, we have $\Pr[p = v^i \suchthat Y] = 0$.
	
	Since $\Pr[p = v \suchthat X] \geq \Pr[p = v \suchthat Y]$ for all $v < \infty$, switching from signal $X$ to signal $Y$ has the effect of transferring the probability of the price being a given value away from the highest values and toward infinity. Moreover, we can deduce which seller types account for this move towards infinity by invoking the well-known result that the seller's pricing schedule is a monotone function of her cost. Let
	$$q^* := \sum_{i = 1}^{f(Y)} \Pr[p = v^i \suchthat Y].$$
	Monotonicity implies that the bottom $q^*$-fraction of seller types behave identically in the old equilibrium given signal $X$ versus signal $Y$, whereas the top $(1 - q^*)$-fraction of seller types set $p = \infty$ under signal $Y$.
	
	Using these structural observations, we will show that the following alternative set of strategies and beliefs constitute an alternative subgame-perfect Bayes-Nash equilibrium, which we denote $E_2$:
	\begin{itemize}
		\item[] \textbf{Strategies:} On all rounds $r \geq 2$, the players play identically in $E_2$ as in $E_1$. On round $r = 1$, the Buyer sends a uniformly random signal. The seller responds as she would in $E_1$ if the buyer had sent signal $X$.
		\item[] \textbf{Beliefs:} Before every message is sent, beliefs are the same in $E_2$ as in $E_1$. After the final buyer message is sent, the seller does not update her beliefs at all.
	\end{itemize}
	To conclude the proof of the theorem, it suffices to verify the following two claims:
	\begin{enumerate}[(a)]
		\item\label{itmNewEquilibriumBetter} All types of both players expect to receive the same payoffs under $E_2$ as they did under $E_1$.
		\item\label{itmNewEquilibriumIsEquilibrium} $E_2$ is a subgame-perfect Bayes-Nash equilibrium.
	\end{enumerate}
	
	We begin with (\ref{itmNewEquilibriumBetter}). For the buyer, we observe that, in the final round of communication under $E_1$, the distribution over prices of any other signal $Y$ stochastically dominates the distribution over prices given $X$. This implies $X$ is weakly better for all buyer types, so his payoff under $E_2$ is weakly greater than it is under $E_1$. But it cannot be strictly greater for any buyer type, for otherwise the buyer would have deviated in $E_1$. Hence, every buyer type receives the same expected utility in $E_2$ as in $E_1$. Likewise, we claim that every seller type is weakly better-off from their response to the $X$ signal under $E_1$. As we have noted, switching from $Y$ to $X$ results in some seller types posting finite prices rather than $\infty$, with nothing else changing. Seller types posting a price of $\infty$ were getting zero utility anyway. Thus, all seller types receive a weakly higher payoff by responding as if the buyer sent signal $X$, as they do under $E_2$. But it cannot be strictly greater for any seller type, for otherwise the seller would have deviated in $E_1$. Hence, every seller type receives the same expected utility in $E_2$ as in $E_1$, concluding the proof of (\ref{itmNewEquilibriumBetter}).
	
	To prove (\ref{itmNewEquilibriumIsEquilibrium}), first observe that $E_2$ is clearly Bayesian consistent given that $E_1$ was. (The only possibly new beliefs that we need to check are after the final message has been sent. This message conveys no information under equilibrium play, so beliefs are not updated.) All that remains to show is that neither player has an incentive to deviate in any subgame.
	
	Fix an arbitrary subgame, and first consider the buyer's incentives. Suppose the buyer achieves expected utility $u$ from his equilibrium strategy $s$ under equilibrium $E_2$. From (\ref{itmNewEquilibriumBetter}), he also achieves utility $u$ under equilibrium play in $E_1$. Consider an alternative strategy $s'$, yielding expected utility $u'$ under $E_2$. Let $s''$ be the same strategy as $s'$, except that the buyer sends signal $X$ in the final round. Clearly, $s''$ yields the same utility $u''$ under $E_1$ as under $E_2$, since the seller responds as if the buyer sent signal $X$ in $E_2$ anyway. Thus, we have
	$$u' \leq u'' \leq u,$$
	where the first inequality follows from the fact that prices under the $X$ response are stochastically dominated by prices under any other response, and the second inequality follows from the fact that $E_1$ is an equilibrium. Thus, $E_2$ is sequentially rational for the buyer.
	
	For the seller, we invoke a similar argument. Suppose she achieves expected utility $u$ from her equilibrium strategy $s$ under equilibrium $E_2$ (and the same utility $u$ under her equilibrium strategy in $E_1$). Consider an alternative strategy $s'$, yielding expected utility $u'$ under $E_2$. Let $s''$ be the same strategy as $s'$ except that, when posting a menu at the end, the seller pretends that the buyer sent a uniformly random message, then responds according to $s'$. At equilibrium $E_2$, $s''$ yields the same distribution over outcomes (since the buyer is already randomizing over messages), and thus the same utility $u'' = u'$ for the seller. However, by its definition, $s''$ has the additional property that it yields the same expected utility for the seller in $E_1$, as the only discrepancy between the buyer's strategy in $E_1$ versus $E_2$ comes at the final round $r = 1$, but $s''$ ignores the buyer's signal anyway. Thus, we have
	$$u' = u'' \leq u,$$
	where the inequality again follows from the fact that $E_1$ is an equilibrium. Thus, $E_2$ is sequentially rational for the seller as well. This concludes the proof.
\end{proof}

\section{Settings Where Cheap Talk Creates Value}\label{secPositive}

We now consider several extensions that go beyond the simple hypotheses of Theorem~\ref{thmImpossibility}. In each, we describe explicit examples in which cheap talk games yield higher profit and/or consumer surplus than in the no-communication game. We verify Several of our results with computations. A subset of those were discovered via computational search over random instances. We begin by discussing this approach.

\subsection{Computational Approach}\label{subLP}

In the absence of cheap talk, it is well known that the seller's optimal menu design problem can be formulated as a linear program, optimizing the expected profit subject to the individual rationality (IR) and incentive compatibility (IC) constraints:\\

\lpmax{$$\sum_{i = 1}^n \Pr[i] \left(p^i - \sum_{j = 1}^m x^i_j c^k_j\right)$$}
{
	$x^i_j \in [0, 1]$ & for all $i \in [n]$, $j \in [m]$\\
	$p^i \geq 0$ & for all $i \in [n]$\\
	$\sum_{j = 1}^m x^i_j v^i_j - p^i \geq 0$ & for all $i \in [n]$ (IR constraint)\\
	$\sum_{j = 1}^m x^i_j v^i_j - p^i \geq \sum_{j = 1}^m x^{i'}_j v^i_j - p^{i'}$ & for all $i, i' \in [n]$ (IC constraint)\\
}\\

This can sometimes yield solutions with fractional allocations, which is a broader class of mechanisms than we study in this paper. However, in all of our examples, the optimal solutions happen to be integral.

We may use this LP as a subroutine to verify pure-strategy equilibria to 1-round cheap talk games as follows. Any such equilibrium induces a partition of the buyer types defined by which signal each type sends. For each part in the partition, and each seller type, we compute the optimal menu using the LP, restricted to the set of types in the part. This gives us a distribution over menus for each possible signal that the buyer can send, where the randomness is over the seller's unknown type. We then verify that each buyer type gets weakly higher expected utility from the distribution over menus associated with that buyer signal than any other signal. If so, we have a valid equilibrium.

\subsection{Multiple Goods, Multiple Units}\label{subMultiUnits}

The aforementioned approach enabled us to search for and verify the existence of cheap talk equilibria that benefit both the buyer and seller. This approach was used to discover the example from Section~\ref{subExampleMechDesign}. Further modifications to the LP enabled us to compute the seller-optimal bilateral trade mechanism yielding the even higher commitment profit.

If we add the constraint
$$x^i_1 \geq x^i_2 \ \ \txt{for all $i \in [n]$}$$
we can represent settings where there are two units of the same good to be sold, and so the first unit (good 1) must be sold before the second unit (good 2). We are thus able to verify that cheap talk can be useful in the instance in Table~\ref{tabUnits}, where the seller has a constant marginal cost per unit, and the buyer has a diminishing marginal cost. The no-communication game confers a seller profit of \$21.50. There is a one-round cheap talk equilibrium in which the buyer reveals their marginal value for the second unit, and the seller profit rises to \$22.5.

\begin{table}[h!]
	\centering
	\begin{tabular}{|r|c c|c c c c|} 
		\hline
		& \multicolumn{2}{c|}{Seller cost (\$)} & \multicolumn{4}{c|}{Buyer value (\$)}  \\		
		& Type 1 & Type 2 & Type 1 & Type 2 & Type 3 & Type 4 \\
		\hline
		First unit & 20 & 60 & 50 & 80 & 40 & 70\\
		
		Second unit & 20 & 60 & 40 & 40 & 10 & 10\\
		\hline
	\end{tabular}
	\caption{An instance with two identical units of a good being sold. All seller/buyer types are equally likely.}
	\label{tabUnits}
\end{table}

\subsection{Repeated Play}\label{subRepeatedPlay}

We now consider a second extension that departs from the basic model: the entire game is infinitely repeated with some discount factor $0 < \delta < 1$. It is well-known that cooperation can be rational and bring welfare gains in such settings, without needing to bring in cheap talk. For instance, the parties could agree to always trade at a fixed price because it is efficient to do so, even if it happens at times to be above the buyer's value or below the seller's cost. To make the setting less trivial, and to distinguish the roles of repeat-play cooperation and cheap talk specifically, we impose an additional constraint: the players' decisions in each round must to be individually rational. No player can reap negative utility in any round.

With this setup, we consider the following instance with a single good. The seller always has cost zero, and the buyer has cost \$4 or \$9, each with probability $\frac12$ (independently drawn across rounds). The discount factor is $\delta = 0.9$.

First note that, in the (repeated) no-communication game, the seller must always set a price of \$9 because that is her optimal price in the one-shot game. The seller may sometimes choose to offer the buyer a lower price, but there is no room for reciprocity because the individual rationality constraint makes it impossible for the buyer to behave any differently than he otherwise would to the benefit the seller.

However, there is an equilibrium of a repeated $(2, 1)$-cheap talk game that yields higher welfare for both parties. On rounds 1, 5, 9, 13, 17, \dots, the buyer reveals his value completely. Then the seller charges that price on the current round and the other price (\$4 or \$9) on the following 3 rounds. (The buyer doesn't meaningfully communicate in those following three rounds.) If any player deviates (including the buyer deciding not to buy even after telling the seller he had a high value), we resort to the equilibrium where the seller always charges \$9.

\textbf{Buyer's incentives.} The only time a buyer could plausibly deviate is by pretending to have a low value when his value is actually high. If so, his utility from the current round will be 5, but will get zero utility on the next three rounds. On the other hand, if the buyer reports honestly, he will achieve zero utility on the current round but a discounted expected utility of
$\$2.5\times(0.9 + 0.9^2 + 0.9^3) = \$6.0975$
over the next three rounds. Thus, the buyer has no desire to deviate.

\textbf{Seller's incentives.} The seller stands to benefit the most from deviating on the first of the three rounds, when she has promised to set price of \$4. By setting a price of \$9 instead she obtains an expected utility of \$4.5 on this round, and on all future rounds. By continuing to play honestly, she obtains utility \$4 from the next three rounds and \$6.5 on the next communication round. On subsequent blocks of 4 rounds, the utilities will be at least this high (they could be higher if the buyer reports a low value). The discounted expected utility per block from playing honestly is thus
$$\$4\times(1+0.9 + 0.9^2) + \$6.5\times(0.9^3) = \$15.5785,$$
and from deviating is
$$\$4.5 \times (1+0.9 + 0.9^2 + 0.9^3) = \$15.4755.$$
Hence, adherence, not deviation, is also optimal for the seller.

Thus, we have shown that both players are better off in the equilibrium of the repeated $(2, 1)$-cheap talk game than in any equilibrium of the repeated no-communication game.

\subsection{Interdependent Values}\label{subInterdependentValues}

Our next extension assumes that the buyer's value depends not only on his own type, but on the seller's type as well, as would happen if there was an element of common value. In other words, we have values $v^{i, k}_j$ rather than $v^i_j$. We also require the seller to disclose this private information to the buyer before closing the sale. (This is a realistic assumption in light of the example from Section~\ref{subExampleBuyerWelfareImprovement}.) Formally, we assume that the buyer learns the seller's type $k$ during step (\ref{itmGameSellerPostsMenu}) of the cheap talk game. As it turns out, this extension can be thought of in terms of our basic model.

\begin{proposition}\label{proReduction}
	For any interdependent values setting (as described above) with $m$ goods and $\ell$ seller types, there is an equivalent independent-values instance with $m \cdot \ell$ goods (and the same numbers of buyer and seller types), in the sense that both instances admit the same equilibrium utilities in $(M, R)$-cheap talk games for any $M$ and $R$.
\end{proposition}

\begin{proof}
	For each good $j$ in the original instance, we create goods $(j, 1), (j, 2), \dots, (j, \ell)$, where $v^i_{(j, k)}$ is defined as $v^{i, k}_{j}$. The seller's cost is infinity (or some suitably high number) for any good $(j, k)$ where $k$ is not the seller's type. Thus the seller only has one such good from the set $\{(j, k) \suchthat k \in [\ell]\}$ for sale, and the buyer's value for it depends on the seller's type in the exact same way as in the original interdependent values instance.
\end{proof}

This proposition gives us a way to verify claims about interdependent values instances, which is how we checked that cheap talk improved both buyer and seller utilities in the Example from Section~\ref{subExampleBuyerWelfareImprovement}.

\subsection{Multiple Rounds of Communication}\label{subMultiRounds}

The examples we have considered thus far involved a single round of communication from the buyer to the seller. Our final result shows that multiple rounds of communication can yield even greater profit for the seller.

First, consider the interdependent values instance, which we can again phrase as a car-purchasing scenario. There are three kinds of buyers. Type 1 buyers will purchase either a red car or blue car for \$90. Type 2 buyers value only red cars, at \$60, and Type 3 buyers value only blue cars, also at \$60. The seller has only one type of car, with cost zero. We call this instance $A$, which we depict in Table~\ref{tabMultiRoundsA} as a multiple-goods instance using the transformation from Proposition~\ref{proReduction}.

\begin{table}[h!]
	\centering
	\begin{tabular}{|r|c c|c c c|} 
		\hline
		& \multicolumn{2}{c|}{Seller cost (\$)} & \multicolumn{3}{c|}{Buyer value (\$)}  \\		
		& Has red & Has blue & Type 1 & Type 2 & Type 3 \\
		\hline
		Red car & 0 & 100 & 90 & 60 & 0\\
		
		Blue car & 100 & 0 & 90 & 0 & 60\\
		\hline
	\end{tabular}
	\caption{Buyer values and seller costs in instance $A$.}
	\label{tabMultiRoundsA}
\end{table}
\vspace{-.7cm}
\begin{table}[h!]
	\centering
	\begin{tabular}{|r|c c c c c c|c c c c c c c c c|} 
		\hline
		& \multicolumn{6}{c|}{Seller cost (\$)} & \multicolumn{9}{c|}{Buyer value (\$)}  \\		
		& 1 & 2 & 3 & 4 & 5 & 6 & 1 & 2 & 3 & 4 & 5 & 6 & 7 & 8 & 9 \\
		\hline
		Good 1 & 0 & 100 & 100 & 100 & 100 & 100 & 90 & 90 & 90 & 60 & 60 & 60 & 0 & 0 & 0\\
		Good 2 & 100 & 0 & 100 & 100 & 100 & 100 & 90 & 90 & 90 & 0 & 0 & 0 & 60 & 60 & 60\\
		Good 3 & 100 & 100 & 0 & 100 & 100 & 100 & 90 & 60 & 0 & 60 & 0 & 90 & 60 & 0 & 90\\
		Good 4 & 100 & 100 & 100 & 0 & 100 & 100 & 90 & 0 & 60 & 0 & 60 & 90 & 0 & 60 & 90\\
		Good 5 & 100 & 100 & 100 & 100 & 0 & 100 & 90 & 0 & 60 & 60 & 90 & 0 & 90 & 0 & 60\\
		Good 6 & 100 & 100 & 100 & 100 & 100 & 0 & 90 & 60 & 0 & 0 & 90 & 60 & 90 & 60 & 0\\
		\hline
	\end{tabular}
	\caption{Buyer values and seller costs in instance $B$.}
	\label{tabMultiRoundsB}
\end{table}

In the no-communication game, the seller faces a buyer whose value is equally likely to be \$0, \$60, or \$90, so she prices at \$60, for a total surplus of $\$60 \cdot \frac23 = \$40$. With one round of cheap talk, she can require the buyer to identify whether or not he is type 3. In this case it turns out that this yields an equilibrium where the seller's expected profit rises to \$45.

Building on this example, now consider instance $B$ in Table~\ref{tabMultiRoundsB}. We have the following result.

\begin{theorem}\label{thmMultiRounds}
	In instance $B$, the greatest profit in a pure-strategy equilibrium of a 1-round cheap talk game is \$42.50, but there is a 2-round cheap talk game with a pure-strategy equilibrium yielding profit \$45.
\end{theorem}

\begin{proof}
	The first claim was computationally verified by enumerating all partitions of the 9 buyer types as described in Section~\ref{subLP}. The optimal partition places either type 4 or type 8 in a singleton part, with all other types together in another part.
	
	The 2-round cheap talk equilibrium sees the seller first tell the buyer whether he is looking to sell either goods 1 or 2, goods 3 or 4, or goods 5 or 6. Assuming the seller speaks honestly, the resulting subgame is equivalent to the 1-round cheap talk game over instance $A$, in which the seller obtains profit \$45. To see that the seller cannot profitably deviate, observe that each of the three subgames uses a different partition for the buyer's signal:
	\begin{itemize}
		\item If the seller is looking to sell goods 1 or 2, the buyer signals according to the partition $\{\{1, 2, 3\} \cup \{4, 5, 6\}, \{7, 8, 9\}\}$.
		\item If the seller is looking to sell goods 3 or 4, the buyer signals according to the partition $\{\{1, 6, 9\} \cup \{2, 4, 7\}, \{3, 5, 8\}\}$.
		\item If the seller is looking to sell goods 5 or 6, the buyer signals according to the partition $\{\{1, 5, 7\} \cup \{3, 4, 9\}, \{2, 6, 8\}\}$.
	\end{itemize}
	(We use unions to denote the merging of types 1 and 2 according to the signaling scheme from instance $A$.)
	One can check that, for each deviation, and every possible signal response, the seller learns nothing. For example, if the seller is type 3 or 4, but pretends to have type 1 or 2, and the buyer reveals his type to be in the set $\{7, 8, 9\}$. The posterior distribution in this case, as in all cases where the seller deviates, puts equal probability on values of \$0, \$60, and \$90---precisely the prior distribution in instance $A$. Thus, the seller will receive an expected profit of \$40 if she deviates, versus \$45 when she plays honestly.
\end{proof}

We remark on one final oddity observed in instance $B$: In the one-round cheap talk equilibrium, the buyer's surplus decreases, and does so by more than the seller's profit increases. Thus the availability of cheap talk actually increases deadweight loss, reducing the total surplus from $\$40 + \$10 = \$50$ to $\$42.50 + \$6.67 = \$49.17$.

\section{Discussion}\label{secDiscussion}

\subsection{Cheap Talk in Practice}\label{subPractice}

Our technical results have demonstrated that cheap talk can be rational in bilateral trade, even under the standard assumptions of fully rational individuals who are steeped in game theory. It can substantially enhance performance across a broad array of realistic situations. We now consider how behavioral regularities can further amplify these gains.

Cheap talk is used constantly in real-world interactions---including in single-message settings where, in principle, it should have no effect. Some of this communication is indeed inconsequential chatter. Yet even statements that game theorists would classify as pure ``babbling'' can influence outcomes once contextual cues and human decision tendencies are taken into account.

Many people, for example, are uncomfortable making claims that feel deceptive. If a buyer says, ``My spouse told me not to offer more than \$500 for this bureau,'' a seller may infer that the statement is likely genuine rather than part of a rehearsed bargaining script. If a substantial fraction of buyers find deception discomforting, the seller has reason to take such remarks seriously. The same applies to a roadside antique store that posts a sign reading ``Prices are not negotiable.'' Most shoppers will treat it as binding and refrain from bargaining, even though an occasional buyer may succeed in negotiating an exception.

Buyers can also transform cheap talk into something closer to ``costly talk'' by taking small steps that enhance credibility. Many of us have exited a car dealership saying some version of: ``If you won’t move closer to my number, I’ll go to your competitor down the road.'' A salesperson will assume a bluff, since bluffing is a common feature when buying a car. But if the buyer produces a business card from the competing dealer or, better, a purchase proposal not yet filled out, the threat becomes much more believable. Adding modest but visible steps---bits of ``effort'' or evidence---elevates the cost of the statement and thus boosts its weight.

Experienced negotiators learn to mimic these cues. Sellers, for instance, commonly promise to ``beat any competitive price,'' a pledge that buyers rarely redeem, but most tend to interpret (or misinterpret) it as a signal that the quoted price is already attractive.

Cultural practices can further endow cheap talk with structure and cost. In many societies, extended pre-negotiation rituals---such as shared meals, gift exchanges, or formal courtesies---impose modest but real costs on entering a negotiation and establish personal connections that raise the psychological cost of misrepresentation, thereby enriching the informational value of verbal communication.

Earlier sections showed that cheap talk becomes directly valuable as soon as we venture a step beyond a single-play setting. Parties who interact repeatedly---such as suppliers and long-term clients---understand the discipline that repetition imposes. A supplier who too often cites ``rising costs'' to justify price increases risks losing steady customers who know from experience that such claims are exaggerated.

Cheap talk also becomes informative when negotiations involve multiple goods or attributes. A speaker may admit a high valuation in one dimension to credibly assert a low valuation in another. A job candidate might say: ``I care most about title and responsibility; I am willing to leave salary where it is, even though I have higher offers elsewhere. If you can make me an associate rather than an assistant vice president, I will accept at your current salary.'' By conceding on one attribute, the candidate gains credibility when claiming high value on another. Importantly, this logic does not depend on behavioral considerations: conceding on one dimension to gain credibility on another would also yield value in a formal model with fully rational players. The behavioral setting simply makes the mechanism more natural and more likely to arise spontaneously.

In short, while cheap talk yields no value in the textbook one-good, one-message, one-round model, real-world actors rarely confine themselves to that model. They introduce small costs, rituals, corroborating signals, or multi-attribute trade-offs that make cheap talk informative. Yet even sophisticated analysts rarely pause to ask why it works; it works often enough that they take the mechanism for granted. In richer settings---exchanges of messages, negotiations over multiple dimensions, or repeated interactions---claims about value and cost become mutually reinforcing and thus credible, even when no statement is directly verifiable.

Often, even when a negotiation consists of a single message followed by a straightforward financial transaction, other dimensions lurk in the background. The biblical account in Genesis 23, describing Abraham's purchase of the Cave of Machpelah and the surrounding field from Ephron the Hittite as a burial site, provides a canonical illustration. (The account is embraced by Islam, Christianity and Judaism.) The land would immediately serve as a tomb for wife Sarah, and later for himself and his descendants.

The negotiation takes place publicly, before Ephron's family, his community, and especially its elders. Following the customary formalities, Ephron addresses Abraham: ``My lord, hearken unto me: the land is worth four hundred shekels of silver; what is that betwixt me and thee? bury therefore thy dead.'' (Genesis 23:15, King James Version.) Abraham offers no counterproposal. He simply accepts the stated price and pays the four hundred shekels of silver.

This brief exchange has attracted extensive commentary from rabbis, philosophers, and modern analysts, including game theorists. There is broad agreement on one central point: Ephron's statement was cheap talk in the formal sense: costless, unverifiable, and nonbinding. Yet the public setting gave it force. Framed as a gesture of courtesy, the remark anchored the price at a high level. Most commentators conclude that Abraham substantially overpaid, but that that overpayment was deliberate. He secured publicly witnessed, unambiguous title, together with the moral legitimacy that came with it. In this way, Ephron's cheap talk facilitated an exchange that yielded significant value to both parties.

This example shows that cheap talk in real-world settings can complement institutions that enhance its impact. Its persistent use in practice is not a curiosity. Rather it reveals a practical tool producing better deals.

\subsection{Future directions}\label{subFutureDirections}

Though this study resolves a number of technical questions, it admittedly opens many more. Little is understood about what structural properties of a given bilateral trade setting render cheap talk useless or more or less valuable. Do fundamental limits exist?

For example, one concrete question is the factor by which cheap talk may improve the welfare of each player. In the example from Section~\ref{subExampleBuyerWelfareImprovement}, cheap talk allows the buyer to quadruple his surplus. This generalizes easily to arbitrary large constant factors, even infinity (i.e., it is possible that the buyer's surplus jumps from zero to nonzero). However, for the seller, we do not even know whether cheap talk can improve the expected profit by a factor of, say, 2.

Beyond enumerating partitions of the buyer's type space, we know of no systematic way to compute optimal cheap talk games for either player. Searching over the space of mixed strategy equilibria in games of imperfect information is notoriously difficult. Our setting does not appear to be any easier. Even the restriction to pure strategy equilibria, the complexity of finding an optimal cheap talk game given discrete, finitely-supported priors remains an open problem. Presumably it could even be PSPACE-Hard. Our Theorem~\ref{thmMultiRounds} indicates multiple rounds of interaction may be required.

A technical extension, not considered in this paper, is risk aversion. Risk aversion becomes relevant only when allocations can be randomized (otherwise the values of goods can be thought of simply as expected utilities and nothing changes). Could there be selling mechanisms, even for a single good, that employ randomized allocations to induce risk-averse buyers to reveal information in a way that facilitates profits?

A more open-ended question is, what is the ``right'' model of mechanism design when the designer holds private information? We have argued for a specific way of capturing the designer's limited commitment power. The designer may commit to a menu of prices to apply at the end of the game.  Before that the only actions players can take are cheap talk messages. While we believe that this format is simple and realistic, we could also allow for other dynamic elements, such as voluntary payments from one party to the other mid-way through the protocol, while receiving subsets of goods or even nothing in return. In any case, we hope to have made the case that the existing canon of ideas from mechanism design hardly settles this question.

Our work calls for a broader consideration of when and where cheap talk may prove useful. It is clearly useless in a zero-sum game. Bilateral trade shares many characteristics of a zero-sum game, in that the seller always wants higher prices and the buyer always wants lower prices. Yet, as we have shown, cheap talk can yield value in surprising ways because there is always an element of common interest when efficiencies loom. One area that is ripe for further exploration is in allocation mechanisms, where agents strategically compete for scarce resources. Even the slightest heterogeneity in preferences could open a window to potential gains if information is effectively shared via cheap talk. Equivalent reasoning could hold that cheap talk generates value in a wide variety of other adversarial games, such as negotiating territorial disputes or troubled partnerships.

The cheap talk concept bristles with implications for how negotiations are and might be conducted. In Section~\ref{subPractice} we made several observations in this realm.  We looked in particular at the intertwining of the technical findings presented here and behavioral proclivities. Empirical studies of the actual use of cheap talk would be welcome. They could come both from data sets, e.g., on dance of negotiation in corporate mergers, and from studies using human subjects in laboratories. 

\subsection{Final Thoughts}

Cheap talk presents a puzzle. It is ubiquitous in economic and social interaction, yet it rarely conveys decisive information, never achieves full efficiency, and only sometimes yields value. This pattern is neither paradoxical nor na\"ive. Through experience and trial and error. rather than equilibrium calculation, individuals retain practices—cheap talk among them—that sometimes reduce missed opportunities to make beneficial exchanges, promote partial alignment of interests, and by doing so reduce deadweight loss. The qualifiers matter. Cheap talk works only sometimes; it can foster alignment without fully securing it; it can mitigate inefficiency without eliminating it. It does not direct play to the optimal path, but it can nudge interaction toward a more favorable one. And as we have also seen, it can also be exploited by one party to the direct detriment of the other.

In this sense, cheap talk belongs to a broader class of low-cost social mechanisms that persist precisely because they can improve outcomes where more formal or verifiable instruments are infeasible, prohibitively costly, or too rigid to deploy. Rational agents should therefore not view cheap talk as a behavioral crutch or a failed substitute for commitment. Rather, it is best understood as a strategic device—one capable of generating value along the path of play even while falling short of full efficiency. Our analysis identifies both the channels through which such value arises and the limits that bound its reach.

\section*{Acknowledgments}

We are deeply grateful to Elliot Lipnowski for his excellent feedback and pointers on an earlier draft of this paper. We also thank Andrzej Skrzypacz and Dirk Bergemann for very helpful discussions. Zeckhauser thanks the Mossavar-Rahmani Center for Business and Government for research support.

\bibliographystyle{plain}
\bibliography{bibliography}

\end{document}